\documentclass[a4paper,UKenglish,cleveref, autoref]{lipics-v2019}
\usepackage{qctl2qbf}

\usepackage{stmaryrd}
\usepackage{tikz}
\usepackage{config-tikz}
\usepackage{xspace}

\usepackage{amsmath}
\usepackage{amssymb}
\usepackage{stackengine}
\usepackage{scalerel}
\usepackage{url}

\usepackage{todonotes}

\title{From Quantified CTL to QBF} 

\author{Akash Hossain}{IRIF, Univ.\ Paris Diderot, France}{akashoss42@gmail.com}{}{ANR FREDDA}

\author{Fran\c{c}ois Laroussinie}{IRIF, Univ.\ Paris Diderot, France}{francoisl@irif.fr}{}{ANR FREDDA}

\authorrunning{A. Hossain and F. Laroussinie}

\Copyright{Akash Hossain and Fran\c{c}ois Laroussinie}



\category{}

\relatedversion{}

\supplement{}

\nolinenumbers 

\hideLIPIcs  

\EventEditors{John Q. Open and Joan R. Access}
\EventNoEds{2}
\EventLongTitle{42nd Conference on Very Important Topics (CVIT 2016)}
\EventShortTitle{CVIT 2016}
\EventAcronym{CVIT}
\EventYear{2016}
\EventDate{December 24--27, 2016}
\EventLocation{Little Whinging, United Kingdom}
\EventLogo{}
\SeriesVolume{42}
\ArticleNo{23}

\begin{document}

\maketitle

\begin{abstract}

\QCTL\ extends the temporal logic \CTL\ with quantifications over atomic propositions. 
This extension is known to be very expressive~: \QCTL\ allows us to express complex properties over Kripke structures (it is  as expressive as \MSO).
Several semantics exist for the quantifications~: 
here, we work with  the \emph{structure semantics},  where  the extra propositions label 
the Kripke structure (and not its execution tree), and the model-checking problem is known to be \PSPACE-complete in this framework. 
We propose a model-checking algorithm for \QCTL\  based on a reduction to \QBF. We consider several reduction strategies, and we compare them with a prototype (based on the SMT-solver Z3) on several examples. 

\end{abstract}


\section{Introduction}

Temporal logics have been introduced in computer science in the late
1970's~\cite{Pnu77}; they~provide a powerful formalism for specifying  correctness properties of
 evolving systems. Various kinds of temporal
logics have been defined, with different expressiveness and
algorithmic properties. For~instance, the~\emph{Computation Tree Logic}~(\CTL)
expresses properties of the computation tree of the system under study (time is branching~: a state may have several successors), and the 
\emph{Linear-time Temporal
Logic}~(\LTL) expresses properties of one execution at a time (a system is viewed as a set of executions).

Temporal logics allow \emph{model checking}, i.e.\ the automatic verification that a finite state system satisfies its expected behavioral specifications~\cite{QS82a,CE82}. It is well known that \CTL~model-checking is \PTIME-complete and \LTL~model-checking (based on automata techniques) is \PSPACE-complete. But verification tools exist for both logics and model-checking is now commonly used in the design of critical reactive systems. The main limitation to this approach is the state-explosion problem~: symbolic techniques (for example with  BDD), SAT-based approaches, or partial order reductions have been developed and they are impressively successful. The SAT-based model-checking consists in using  SAT-solvers in the decision procedures. It was first developed for bounded model-checking (to search for executions whose length is bounded  by some integer,  satisfying some temporal property) which can be reduced to some satisfiability problem and then can be solved by a SAT-solver~\cite{BiereCCZ99}. SAT approaches have also been extended to unbounded verification and combined with other techniques~\cite{McMillan02}. Many studies have been done in this area, and it is widely considered as  an important approach in practice, which complements other symbolic techniques like BDD (see~\cite{BiereCCSZ03} for a survey).

In~terms of expressiveness, \CTL\ (or  \LTL) still has some limitations~: in~particular,
it~lacks the ability of \emph{counting}. For instance, it~cannot express that
an event occurs (at~least) at every even position along a path, or that a
state has two successors. In~order to cope with this, temporal logics have
been extended with \emph{propositional quantifiers}~\cite{sis83}~: those
quantifiers allow for adding fresh atomic propositions in the model before
evaluating the truth value of a temporal-logic formula. That~a state has at
least two successors can then be expressed (in~\emph{quantified \CTL},
hereafter written~\QCTL) by saying that it~is possible to label the model with
atomic proposition~$p$ in such a way that there is a successor that is labelled
with~$p$ and one that is~not. 

Different semantics for \QCTL\ have been studied in literature depending on the definition of the labelling~: either it refers to the finite-state model -- it is the \emph{structure} semantics -- or it refers to the execution tree -- it is the \emph{tree semantics}. 
Both semantics are interesting and have been extensively studied~\cite{Kup95a,Fre01,PBDDC02,Fre03,DLM12,LM14}.
While the tree semantics allow us to use the tree automata techniques to get decision procedures (model-checking and satisfiability are \TOWER-complete~\cite{LM14}), the situation is quite different for the structure semantics~: in this framework, model-checking is \PSPACE-complete and satisfiability is undecidable~\cite{Fre01}.

In this paper, we focus  on the structure semantics, and we propose a model-checking algorithm based on a reduction to \QBF~: given a Kripke structure $\calK$ and a \QCTL\ formula $\Phi$, we show how to build a \QBF\ formula $\widehat{\Phi}^\calK$ which is valid iff $\calK\sat \Phi$.  It is  natural to use \QBF\  quantifiers  to deal with propositional  quantifiers of \QCTL. Of course, \QBF-solvers are not as efficient as SAT-solvers, but still much progress has been made   (and \QBF-solvers have already been considered for model-checking, as in~\cite{DershowitzHK05}).
We propose several reductions depending on the way of dealing with nested temporal modalities, and 
we have implemented a prototype (based on Z3 SMT-solver~\cite{MouraB08}) to compare these reductions over  several examples. As far as we know, it is the first implementation of a model-checker for \QCTL. 

Here, our first objective is  to use the \QBF-solver as a tool to check complex properties over limited size models, and this is therefore different from the classical use of  \SAT-based techniques which are precisely applied to solve verification problems for very large systems.  

The outline of the paper is as follows~: we~begin with setting up the necessary
formalism in order to define \QCTL. We~then devote
Section~\ref{sec-reduc} to the different reductions to \QBF. Finally,
Section~\ref{sec-experiments} contains several practical results and examples.


\section{Definitions}

\subsection{Kripke structures}

Let $\AP$ be a set of atomic propositions.

\begin{definition}
A~\emph{Kripke structure} is a tuple $\calK=\tuple{V,E,\ell}$, where $V$~is a
finite set of vertices (or states), $E\subseteq V\times V$ is a set of edges
(we assume  that for any~$v\in V$, there exists $v'\in V$ s.t. $(v,v')\in E$),
and $\ell\colon V\to 2^{\AP}$ is a labelling function.
\end{definition}

An infinite  path (also called an execution) in a Kripke structure is an infinite  sequence $\rho = x_0 x_1 x_2 \ldots$  such that for any $i$ we have $x_i \in V$ and $(x_i,x_{i+1})\in E$. We~write  $\Path^\omega_\calK$ for the
set of infinite paths of~$\calK$ and $\Path^\omega_\calK(x)$ for the set of infinite paths issued from $x\in V$.
Given such a path $\rho$, we use $\rho_{\leq i}$ to denote the $i$-th prefix $x_0\ldots x_i$, $\rho_{\geq i}$ for the $i$-th suffix $x_i x_{i+1}\ldots$, and $\rho(i)$ for the vertex $x_i$. The size of $\calK$ is $|V|+|E|$.

 Given a set $P\subseteq \AP$, two  Kripke structures $\calK=(V,E,\ell)$ and $\calK'=(V',E',\ell')$ are said $P$-equivalent (denoted by $\calK \equiv_P \calK'$) if $V=V'$, $E=E'$, and for every $x \in V$ we have~: $\ell(x)\cap P = \ell'(x) \cap P$.

\subsection{\QCTL}
This section is devoted to the definition of the logic \QCTL ,\ which extends the classical branching-time temporal logic \CTL\ with  quantifications over atomic propositions.   

\begin{definition}
\label{def-QCTL}
The syntax of $\QCTL$ is defined by the following grammar~:
\begin{xalignat*}1
\QCTL \ni \phi ,\psi  &\coloncolonequals  q \mid \neg\phi \mid \phi \ou\psi  
    \mid \Ex \X \phi  \mid \Ex \phi \Until \psi \mid \All \phi \Until \psi  \mid \exists p.\ \phi 
\end{xalignat*}
where $q$ and~$p$ range over~$\AP$.
\end{definition}

\QCTL\  formulas are evaluated over states of  Kripke structures~:

\begin{definition}
Let $\calK=\tuple{V,E,\ell}$ be a Kripke structure, and $x \in V$. 
The semantics of \QCTL\ formulas is defined inductively as follows~: 
\begin{align*}
\calK, x \models & p  \text{ iff }  p\in \ell(x) \\
\calK,  x \models & \neg\phi  \text{ iff }  \calK,  x \not\models \phi \\
\calK,  x \models & \phi \ou \psi  \text{ iff }  \calK,  x
\models \phi \text{ or }  \calK,  x \models \psi \\
\noalign{\pagebreak[2]}
\calK,  x \models & \Ex \X \phi   \text{ iff }  \exists (x,x') \in E \text{ s.t.\ } \calK, x' \models \phi \\
\calK,  x \models & \Ex \phi  \Until \psi  \text{ iff }  \exists  \rho  \in \Path^\omega_\calK(x),  
  \exists i\geq 0  \text{ s.t.\ }  \calK,  \rho(i) \models \psi \text{ and } \\
   & \qquad \qquad \qquad \qquad \qquad \text{ for any }  0 \leq j <i, \text{ we have }  \calK, \rho(j) \models \phi  \\
\calK,  x \models &\All \phi  \Until \psi  \text{ iff }  \forall  \rho  \in \Path^\omega_\calK(x),  
  \exists i\geq 0  \text{ s.t.\ }  \calK,  \rho(i) \models \psi \text{ and }  \\
  & \qquad \qquad \qquad  \qquad\qquad  \text{ for any }  0 \leq j <i, \text{ we have }  \calK, \rho(j) \models \phi  \\
  \calK,  x   \models & \exists p.\ \phi  \text{ iff } \exists \calK'
\equiv_{\AP\backslash\{p\}} \calK  \text{ s.t. } \calK', x \models \phi 
\end{align*}
\end{definition}

In~the sequel, we use standard abbreviations such as $\top$, $\bot$, $\et$, $\impl$ and $\equivalent$.
We also use the additional  (classical)   temporal modalities  of \CTL~: $\All\X\phi \eqdef \non \Ex \X \non \phi$
, $\Ex\F\phi \eqdef \Ex \top \Until \phi$, $\All\F\phi \eqdef \All \top \Until \phi$, 
$\Ex \G\phi \eqdef \neg\All \F\neg\phi$,  $\All \G\phi \eqdef \neg\Ex \F\neg\phi$,  $\Ex \vfi \WUntil \psi \eqdef \non \All \non \psi \Until (\non \psi \et \non \vfi)$ and $\All \vfi \WUntil \psi \eqdef \non \Ex \non \psi \Until (\non \psi \et \non \vfi)$. 

Moreover, we use the following abbreviations related to quantifiers over atomic  propositions~:   $\forall p.\ \phi \eqdef \neg\exists p.\ \neg\phi$,  and for a set $P = \{p_1,\ldots, p_k\} \subseteq \AP$, we write $\exists P. \phi$ for  $\exists p_1.  \ldots \exists p_k. \phi$ and $\forall P. \phi$ for  $\forall p_1. \ldots \forall p_k. \phi$.
 
The~size of a formula~$\phi\in\QCTL$, denoted $\size\phi$, is defined inductively by~: $\size q = 1$, $\size{ \non \phi} = \size{\exists p.\phi} = \size{\Ex\X \phi} =  1 + \size \phi$, $\size{\phi \ou \psi} = \size{\Ex \phi \Until \psi}= \size{\All \phi \Until \psi}= 1 + \size \phi + \size \psi$.  
 Moreover we use $\HT(\vfi)$ to denote the  \emph{temporal height} of  $\vfi$, that is the maximal number of nested temporal modalities in $\vfi$. And given a subformula $\psi$ in $\Phi$, the \emph{temporal depth} of $\psi$ in $\Phi$ (denoted $\tdepth_\Phi(\psi)$)  is the number of temporal modalities having $\vfi$ in their scope.
 
In the following, we denote by $\SubF(\Phi)$ (resp.\ $\SubTF(\Phi)$) the set of subformulas of $\Phi$ 
(resp.\ the set of subformulas starting with a temporal modality).

Two \QCTL\ formulas $\vfi$ and $\psi$ are said to be \emph{equivalent} (written $\vfi \equiv \psi$) iff for any structure $\calK$, any state $x$, we have $\calK, x \sat \vfi$ iff   $\calK, x \sat \psi$. This equivalence is substitutive.

\paragraph{Discussion on the semantics.}

The semantics we defined is classically called the \emph{structure semantics}~: a formula $\exists p.\phi$ holds true in a Kripke structure $\calK$ iff there exists a $p$-labeling of the structure $\calK$ such that $\phi$ is satisfied. 
 Another  well-known semantics coexists in the literature for propositional
quantifiers,  the  \emph{tree semantics}~: $\exists p.\ \phi$ holds
 true when there exists a $p$-labelling of the \emph{execution tree} (the infinite unfolding) of the
 Kripke structure under which $\phi$~holds. If, for \CTL, interpreting formulas over the structure or the execution tree is equivalent, this is not the case for \QCTL. Moreover,   these two semantics do not have the same algorithmic  properties~: if \QCTL\ model-checking and satisfiability are \TOWER-complete for the tree semantics (the algorithms are based on tree automata techniques), \QCTL\ model-checking  is \PSPACE-complete for the structure semantics but satisfiability is undecidable.
   (see~\cite{LM14} for a survey). Nevertheless, in both semantics, \QCTL\ and \QCTLs\ (the  extension of \CTLs\ with quantifications) are equally expressive, and are as expressive~\footnote{This requires adequate definitions,
since a temporal logic formula may only deal with the reachable part of the
model, while \MSO\ has a more \emph{global} point of view.} as the Monadic Second-Order Logic over the finite structure or the infinite trees (depending on the semantics). Note also that any \QCTL\ formula is equivalent to a formula in Prenex normal form (we will use this result in next sections). 

Finally, there is also  the \emph{amorphous semantics}~\cite{Fre01}, where
$\exists p.\ \phi$ holds true at a state~$s$ in some Kripke structure~$\calK$
if, and only if, there exists some Kripke structure~$\calK'$ with a state $s'$
such that $s$ and $s'$ are bisimilar, and for which there exists a $p$-labeling
making $\phi$ hold true at~$s'$. With these semantics, the logic is insensitive
to unwinding, and more generally it is bisimulation-invariant (contrary to the two previous semantics, see below). 

\subsection{Examples of \QCTL\ formulas}

\QCTL\ allows us to express complex properties over Kripke structures~: for example, it is easy to build a characteristic formula (up to isomorphism) of a structure, and one can also reduce model-checking problems for multi-player games to \QCTL\ model-checking~\cite{LaroussinieM15}\ldots Below, we give several examples of counting properties, to illustrate the expressive power of propositional quantifiers.

The first one of the   formulas below expresses that there exists a unique reachable state verifying $\vfi$, and the second one states that there exists a unique immediate successor satisfying $\vfi$~:
\begin{align}
\Ex_{=1} \F(\vfi) &\;\eqdef\; \EF (\vfi) \et \forall p. \big(\EF (p\et \vfi) \impl \AG (\vfi \impl p)\big)  \label{FuniqF} \\
\Ex_{=1}\X\phi  &\;\eqdef\; \Ex\X\phi \et \forall p.\ \big(\All\X(\phi \impl p) \ou \All\X(\phi \impl \non p)\big) \label{FuniqX}
\end{align}
where we assume that $p$ does not appear in~$\vfi$. Consider the formula~\ref{FuniqF}:  if there were two reachable states satisfying $\vfi$, then labelling only one of them with~$p$ would falsify the $\AG$ subformula. For \ref{FuniqX}, the argument is similar.

The existence of 
at least $k$ successors satisfying a given property can be expressed with~:
\begin{align}
\Ex_{\geq k}\X\phi = & \exists P.\ 
\Bigl(\ET_{1\leq i\leq k} \Ex\X \bigl(p_i\et \ET_{i'\not= i} \non p_{i'}\bigr) \et
\All\X\Bigl(\bigl(\OU_{1\leq i\leq k} p_i\bigr) \impl \phi\Bigr)\Bigr) 
\end{align}

And we can define $\Ex_{=k} \X \phi$ as to $\Ex_{\geq k} \X \phi  \:\et\: \non \Ex_{\geq k+1} \X \phi$. Note that
these examples show why \QCTL\ formulas are not  bisimulation-invariant. 

When using \QCTL\ to specify properties, one often needs to quantify (existentially or universally) over \emph{one} reachable state we want to mark with a given atomic proposition. To this aim, we add the following abbreviations~: 
\[
\exists^1p.\vfi \eqdef  \exists p. \big( (\Ex_{=1}\F \: p) \et \vfi\big) \quad\quad\quad
\forall^1p.\vfi \eqdef  \forall p. \big( (\Ex_{=1}\F \:  p) \impl \vfi\big)
\]


\section{Model-checking QCTL}
\label{sec-reduc}

Model-checking \QCTL\ is a \PSPACE-complete problem (for detailed results about  program complexity and formula complexity, see~\cite{LM14}), and it is \NP-complete for the  restricted set of formulas of the form $\exists P.\vfi$, with $P\subseteq \AP$ and $\vfi \in \CTL$\cite{Kup95a}.
In this section, we give a reduction from the \QCTL\ model-checking problem to \QBF.

In the following, we assume a Kripke structure $\calK=\tuple{V,E,\ell}$, an initial state $x_0 \in V$ and a \QCTL\ formula $\Phi$ to be fixed. Let $V$ be $\{v_1,\ldots,v_n\}$. 
We also assume w.l.o.g.\  that every quantifier $\exists$ and $\forall$ in $\Phi$ introduces a fresh atomic proposition, and distinct from the propositions used in $\calK$. We use $\AP_Q^\Phi$ to denote the set of  quantified atomic propositions in $\Phi$. 

These  assumptions allow us to use an alternative notation for the semantics of $\Phi$-subformulas~: the truth value of $\vfi$ will be defined for a state $x$ in  $\calK$ within an \emph{environment} $\varepsilon:\AP_Q^\Phi\mapsto 2^V$, that is a partial mapping  associating a subset of vertices to a proposition in $\AP_Q^\Phi$. We use $\calK,x \sat_\varepsilon \vfi$ to denote that $\vfi$ holds at $x$ in $\calK$ within $\varepsilon$.
This ensures that the $\calK$'s labelling $\ell$ is not modified when a subformula is evaluated, only $\varepsilon$ is extended with labellings for  new quantified propositions. 
Formally the main changes of the semantics are as follows~:
\begin{align*}
\calK, x \models_\varepsilon p & \text{ iff }   \begin{cases}  \top & \text{if}\:  (p \in \AP_Q^\Phi \et x \in \varepsilon(p))  \ou (p \not\in \AP_Q^\Phi \et p \in \ell(x)) \\ \bot & \text{otherwise} \end{cases} \\
\calK,  x   \models_\varepsilon  \exists p.\ \phi & \text{ iff } \exists V' \subseteq V  \text{ s.t. } \calK, x \models_{\varepsilon[p \mapsto V']} \phi 
\end{align*}
where $\varepsilon[p \mapsto V']$ denotes the mapping    which coincides with $\varepsilon$ for every proposition in $AP_Q^\Phi\setminus \{p\}$ and associates $V'$ to $p$. 

We use this new notation in order to better distinguish  initial $\calK$'s propositions and quantified propositions to make proofs simpler. Of course, there is no semantic difference~: $\calK,x  \models \Phi$ iff $\calK, x \models_\emptyset \Phi$. 

\medskip

In next sections, we consider general \emph{quantified propositional formulas} (\QBF) of the form~:
\[
\QBF \ni \alpha, \beta \coloncolonequals  q \mid \alpha \et \beta \mid \alpha \ou \beta \mid  \non \alpha \mid \alpha \equivaut \beta \mid \alpha \impl \beta \mid \exists q.\alpha \mid \forall q.\alpha 
\]
The formal semantics of  a formula $\alpha$ is defined over a Boolean valuation for free variables in $\alpha$ (\ie\ propositions which are not bounded by a quantifier~\footnote{We assume w.l.o.g.\ that every quantifier $\exists$ and $\forall$ introduces a new proposition.}), and it is defined as usual. A  formula  is said to be closed when it does not contain free variables. In the following, we use the standard notion of validity for closed \QBF\ formulas.

Our aim is then to build a (closed) \QBF\ formula $\widehat{\Phi}^{x_0}$ such that   $\widehat{\Phi}^{x_0}$  is valid iff $\Phi$ holds true at $x_0$ in $\calK$.

\subsection{Overview}

We present several reductions of \QCTL\ model-checking problem to \QBF\ validity problem. 
Given a $\Phi$-subformula $\vfi$, a vertex $x \in V$, and a subset $P \subseteq AP_Q^\Phi$, we define a \QBF\ formula $\widehat{\vfi}^{x,P}$ whose variables belongs to $\AP_Q^\Phi\times V$ (in the following, we use the notation $p^x$ for  $p\in \AP_Q^\Phi$ and  $x\in V$). 
The first two reductions are based on different encodings of temporal modalities, but share a common part given at Table~\ref{reduc-main}.

\begin{table}
\begin{align*}
\widehat{\non \vfi}^{x,P} &\eqdef \non \widehat{\vfi}^{x,P} \quad \quad 
\widehat{\vfi \ou \psi}^{x,P} \eqdef  \widehat{\vfi}^{x,P} \ou \widehat{\psi}^{x,P} \quad \quad 
\widehat{\EX  \vfi}^{x,P} \eqdef   \OU_{(x,x') \in E} \widehat{\vfi}^{x',P} \\
\widehat{\exists p. \vfi}^{x,P} &\eqdef \exists p^{v_1} \ldots p^{v_n}.  \widehat{\vfi}^{x,P\cup\{p\}} \quad\quad\quad
\widehat{p}^{x,P} \eqdef \begin{cases} p^x & \text{if}\: p \in P \\ \top & \text{if}\:  p \not\in P  \et p \in \ell(x) \\ 
\bot & \text{otherwise} \\ \end{cases} \\ 
\end{align*}
\caption{Reduction for basic modalities  for reductions~\MetUU\ and ~\MetFP}
\label{reduc-main}
\end{table}

\subsubsection{Unfolding characterization of the until operators}

First, we can complete  previous construction rules of  Table~\ref{reduc-main}  with those of Table~\ref{compl1} to get the first method (called \MetUU). This is a  naive approach consisting in encoding  the temporal modalities as unfoldings of the transition relation.

Before stating the correctness of the construction, we need to associate a Boolean valuation $v_\varepsilon$ for variables in $\AP_Q^\Phi\times V$ to an environment $\varepsilon$ for $\AP_Q^\Phi$.  We define $v_\varepsilon$ as follows~: for any $p\in \AP_Q^\Phi$ and $x \in V$,  $v_\varepsilon(p^x) = \top$ iff $x \in \epsilon(p)$.   

\medskip

Now we have the following theorem whose proof is in appendix~\ref{app-uu}~:
\begin{theorem}
\label{theo-UU}
Given a \QCTL\ formula $\Phi$, a Kripke structure    $\calK=\tuple{V,E,\ell}$, a state $x \in V$, an environment $\varepsilon : \AP_Q^\Phi \mapsto 2^V$ and a  $\Phi$-subformula $\vfi$, if  $\widehat{\vfi}^{x,\dom(\varepsilon)}$ is defined inductively w.r.t.\ the rules of Tables~\ref{reduc-main} and~\ref{compl1}, we have:
$\calK, x \sat_\varepsilon \vfi \quad\mbox{iff}\quad v_\varepsilon \sat \widehat{\vfi}^{x,\dom(\varepsilon)}$
\end{theorem}

\begin{table}
\begin{align}
\widehat{\Ex \vfi \Until \psi}^{x,P} &\eqdef   \overline{\Ex \vfi \Until \psi}^{x,P,\{x\}} \quad \quad \text{with:} \\
\overline{\Ex \vfi \Until \psi}^{x,P,X} &\eqdef \; \widehat{\psi}^{x,P} \ou \Big(\widehat{\vfi}^{x,P} \et \OU_{(x,x')\in E s.t. x'\not\in X} \overline{\Ex \vfi \Until \psi}^{x',P,X\cup\{x'\}} \Big)  \\
\widehat{\All \vfi \Until \psi}^{x,P} &\eqdef   \overline{\All \vfi \Until \psi}^{x,P,\{x\}} \quad \quad \text{with:} \\
\overline{\All \vfi \Until \psi}^{x,P,X} &\eqdef \; \begin{cases} \widehat{\psi}^{x,P}  & \text{if}\: \exists (x,x') \in E \et x'\in X\\ 
 {\displaystyle \widehat{\psi}^{x,P}  \ou \Big(\widehat{\vfi}^{x,P} \et \ET_{(x,x')\in E} \overline{\All \vfi \Until \psi}^{x',P,X\cup\{x'\}} \Big)} & \text{otherwise} \\ \end{cases}
\end{align}
\caption{Reduction for temporal modalities  $\EU$ and $\AU$ -- variant \MetUU}
\label{compl1}
\end{table}

It remains to define $\widehat{\Phi}^x$ as  $\widehat{\Phi}^{x,\emptyset}$ and we get the reduction~: $\calK, x_0 \sat  \Phi$ iff $\widehat{\Phi}^{x_0}$ is valid.

The main drawback of this reduction is the size of the \QBF\ formula~:  indeed any Until modality may induce a formula whose size is in $O(|V|!)$, and the size of the resulting formula $\widehat{\Phi}^{x_0}$ 
is then  in $O((|\Phi|\cdot|V|!)^{\HT(\Phi)})$. Nevertheless, one can notice that the reduction does not use new quantified propositons to encode the temporal modalities, contrary to other methods we will see later. 

\subsection{Fixed point characterization of the until operators}

Here we present the fixed point method (called \MetFP) for dealing with the modalities $\All\Until$\ and $\Ex\Until$. Let $\phi$ and $\psi$ be two \QCTL-formulas. The idea of the method is to build a \QCTL\  formula that is equivalent to $\Ex\phi\Until\psi$ (or $\All\phi\Until\psi$), using only the modalities  $\Ex\X$, $\All\X$ and  $\All\G$. 

We first have the following lemma~:

\begin{lemma}
\label{lemFP1}
For any \QCTL\ formula $\Ex \vfi  \Until \psi$, we have~:
\[
\Ex \vfi  \Until \psi \quad\equiv\quad \forall z. \Big( \AG \big( z \equivalent (\psi \ou (\vfi \et \EX \: z)) \big) \; \impl \; z\Big)
\]
\end{lemma}
\begin{proof}
Let $x$ be a state in a Kripke structure $\calK$. Let $\theta$ be the formula $\big(\AG \big( z \equivalent (\psi \ou (\vfi \et \EX \: z)) \big) \; \impl \; z\big)$. \\
Assume $\calK, x \sat \Ex \vfi  \Until \psi$. We can use the standard characterization of $\EU$ as fixed point~: $x$ belongs to the least fixed point of the  equation $Z \eqdef \psi \ou (\vfi \et \EX \: Z)$ where $\psi$ (resp.\ $\vfi$) is here interpreted as the set of states satisfying $\psi$ (resp.\ $\vfi$). Therefore any $z$-labelling of reachable states from $x$~\footnote{Labelling other states does not matter.} corresponding to a fixed point will have the state $x$ labelled. This is precisely what is specified by the \QCTL\ formula. \\
Now if $\calK, x \sat \theta$ for every $z$-labelling corresponding to a fixed point of the previous equation, this is the case for the $z$-labelling of the states reachable from $x$ and satisfying $\Ex \vfi  \Until \psi$, and we deduce $\calK, x \sat \Ex \vfi  \Until \psi$.  
\end{proof}

And we have the same result for $\AU$~:

\begin{lemma}
\label{lemFP2}
For any \QCTL\ formula $\All \vfi  \Until \psi$, we have~:
\[
\All \vfi  \Until \psi \quad\equiv\quad \forall z. \Big( \AG \big( z \equivalent (\psi \ou (\vfi \et \AX \: z)) \big) \; \impl \; z\Big)
\]
\end{lemma}
\begin{proof}
Similar to the proof of Lemma~\ref{lemFP1}
\end{proof}

As a direct consequence, we get the following result~:
\begin{proposition}
\label{prop-qctlres}
For any \QCTL\ formula  $\Phi$, we can build an  equivalent  \QCTL\ formula $\FPC(\Phi)$ such that~:
(1) $\FPC(\Phi)$ is built up from atomic propositions, Boolean operators, propositional quantifiers and modalities $\EX$ and $\AG$, and 
(2) the size of   $\FPC(\Phi)$ is linear in $|\Phi|$.
\end{proposition}
Note that the size of $\FPC(\Phi)$ comes from the fact that there is no duplication of subformulas when applying the transformation rules based on equivalences of Lemmas~\ref{lemFP1} and~\ref{lemFP2}. Moreover, the temporal height of $\FPC(\Phi)$ is $\HT(\Phi)+1$.

And to translate $\FPC(\Phi)$ into \QBF, it remains to add a single rule~\footnote{For $\AX$ we can either change the rule for $\EX$ with a disjunction, or express it with $\EX$ and $\non$.} to the definitions of Table~\ref{reduc-main} to deal with $\AG$~:
\begin{align}
\widehat{\AG \vfi}^{x,P} &\eqdef  \ET_{(x,y)\in E^*}  \widehat{\:\vfi\:}^{\:y,P} 
\label{rule-ag-fp}
\end{align}

And we have~: 
\begin{theorem}
Given a \QCTL\ formula $\Phi$, a Kripke structure    $\calK=\tuple{V,E,\ell}$, a state $x \in V$ and an environment $\varepsilon : \AP_Q^\Phi \mapsto 2^V$, for any   $\Phi$-subformula $\vfi$, if  $\widehat{\FPC(\vfi)}^{x,\dom(\varepsilon)}$ is defined inductively w.r.t.\ the rules of Table~\ref{reduc-main} and the rule~\ref{rule-ag-fp}, we have~:
\[
\calK, x \sat_\varepsilon \vfi \quad\mbox{iff}\quad v_\varepsilon \sat \widehat{\FPC(\vfi)}^{x,\dom(\varepsilon)}
\]
\end{theorem}

\begin{proof}
The proof is a direct consequence of  Proposition~\ref{prop-qctlres}.
\end{proof}

With this approach, the size of $\widehat{\FPC(\Phi)}^{x}$ is in $O((|\Phi|\cdot|\calK|)^{\HT(\Phi)})$. Indeed, an  $\EU$ (or $\AU$) modality gives rise to a  \QBF\ formula of size $(|V|+|E|)$.  
The exponential size comes from the potential nesting of temporal modalities~: to avoid it, one could consider the DAG-size of formulas. In the next section, we consider another solution.  Note also that the number of propositional variables in the \QBF\ formula is bounded by $|\AP_Q^\Phi|\cdot|V|$.

\subsection{Reduction via flat formulas}

To avoid the size explosion of $\widehat{\Phi}^x$, one can use an alternative approach for Prenex \QCTL\ formulas.
Remember that any \QCTL\ formula can be translated into an equivalent  \QCTL\ formula in Prenex normal form whose size is linear in the size of the  original formula~\cite{LM14}. 

In the following, a \CTL\ formula is said to be \emph{basic} when it is of the form $\EX \alpha$, $\Ex \alpha \Until \beta$ or $\All \alpha \Until \beta$ where $\alpha$ and $\beta$ are Boolean combinations of atomic propositions. 
It is easy to observe that any \CTL\ formula can be translated into a  \QCTL\ formula with a temporal height less or equal to 2~:
\begin{proposition}
\label{prop-ctl-flat}
Any  \CTL\ formula $\Phi$ is equivalent to some \QCTL\ formula $\Psi$ of the form~:
\[
\exists \{\kappa_1\ldots \kappa_m\}. \Big( \Phi_0 \et \ET_{i=1\ldots m} \AG (\kappa_i \Leftrightarrow   \theta_i)\Big)
\]
where $\Phi_0$ is a Boolean combination of basic \CTL\ formulas and every $\theta_i$ is a basic \CTL\  formula (for any $1\leq i \leq m$).  Moreover, $|\Psi|$ is in $O(|\Phi|)$. 
\end{proposition}
 
\begin{proof}
Let $S$ be the set of \emph{temporal} subformulas occurring in $\Phi$ at a temporal depth greater or equal to 1. Assume $S \eqdef \{\vfi_1, \ldots, \vfi_m\}$ and let $\{\kappa_1,\ldots,\kappa_m\}$ be $m$ fresh atomic propositions. For any $1\leq i \leq m$, we define $\theta_i$ as $\vfi_i$ where every  temporal subformula  $\vfi_j$ at temporal depth 1 in $\vfi_i$, is replaced by $\kappa_j$. $\Phi_0$ is  defined similarly from $\Phi$.

The equivalence  $\Phi \equiv \Psi$ can be proven by induction over $|S|$. If $|S|=0$, the formula satisfies the property. Now consider a formula with $|S|=m>0$. Let $\vfi_m \in S$ be a basic \CTL\ formula.  Let $\Phi'$ be the following formula~: 
$\exists \kappa_m.(\Phi[\vfi_m\leftarrow \kappa_m] \et \AG (\kappa_m  \equivalent \vfi_m))$, where $\vfi[\alpha\leftarrow \beta]$ is $\vfi$ where every occurrence of $\alpha$ is replaced by $\beta$. 
Then $\Phi$ and $\Phi'$ are equivalent~: any state reachable from the current state $x$ will be labeled by $\kappa_m$ iff $\vfi_m$ holds true at that state (NB~: the states that are not reachable from $x$ do not matter for the truth value of $\Phi$) and this enforces the equivalence.  
And then we can apply the induction hypothesis to get~:
\begin{align*}
\Phi & \equiv  \exists \kappa_m.(\Phi[\vfi_m\leftarrow \kappa_m] \et \AG (\kappa_m  \equivalent \vfi_m)) \\
  & \stackrel{i.h.}{\equiv} \exists\kappa_{m}.\Big[  \Big(\exists \{\kappa_1\ldots \kappa_{m-1}\}. \Big( \Phi_0 \et \ET_{i=1\ldots m-1} \AG (\kappa_i \Leftrightarrow   \theta_i) \Big) \Big) \et \AG (\kappa_m  \equivalent \vfi_m)\Big] \\
  & \equiv  \exists \{\kappa_1\ldots \kappa_m\}. \Big( \Phi_0 \et \ET_{i=1\ldots m} \AG (\kappa_i \Leftrightarrow   \theta_i) \Big)
 \end{align*}
Note that the last    equivalence comes from the fact that no $\kappa_i$ with  $i<m$ occurs in  $\vfi_m$ and we also have  $\theta_m=\vfi_m$.  
The size of $\Psi$ is linear in $|\Phi|$. 
\end{proof}

In the following, a \QCTL\ formula of the form ${\displaystyle \calQ. \big( \Phi_0 \et \ET_{i=1\ldots m} \AG (\kappa_i \Leftrightarrow   \theta_i)\big)}$, where $\calQ$ is a sequence of  quantifications,  $\Phi_0$ is a Boolean combination of basic \CTL\ formulas, and every $\theta_i$ (for $i=1,\ldots,m$) is a basic \CTL\ formula, is said to be a \emph{flat} formula. 
 
As a corollary of previous results, we have~:
\begin{proposition}
\label{prop-flat}
Any \QCTL\ formula $\Phi$ is equivalent to some flat formula whose size is linear in   $|\Phi|$. 
\end{proposition}

Given a \QCTL\ formula $\Phi$, we use $\Flat(\Phi)$ to denote the flat formula equivalent to $\Phi$, obtained by first translating $\Phi$ into Prenex normal form, and then transform the \CTL\ subformula as described in Proposition~\ref{prop-ctl-flat}.

Applying method \MetFP\ to some flat  formula provides a \QBF\ formula of polynomial size since a flat formula has a temporal height less or equal to 2~:
\begin{corollary}
Given a \QCTL\ formula $\Phi$,  a Kripke structure $\calK=\tuple{V,E,\ell}$ and a state $x$, the \QBF\ formula 
\stackon[-8pt]{$ \FPC(\Flat(\Phi))$}{\vstretch{1.5}{\hstretch{8}{\widehat{\phantom{\;}}}}}$^{x}$
obtained by applying the rules of Table~\ref{reduc-main} and the rule~\ref{rule-ag-fp} is valid iff $\calK, x \sat \Phi$. 
And the size of $\stackon[-8pt]{$ \FPC(\Flat(\Phi))$}{\vstretch{1.5}{\hstretch{8}{\widehat{\phantom{\;}}}}}^{x}$ 
 is in $O(|V|\cdot(|V|+|E|)\cdot |\Phi|)$. 
\end{corollary}

Therefore this reduction (called~\MetFPF)  provides a \PSPACE\ algorithm for   \QCTL\ model-checking. 
But there are two disadvantages of this approach. First, putting the formula into Prenex normal form may increase the number of quantified atomic propositions and the number of alternations (which is \emph{in fine} linear in the number of quantifiers in the original formula)\cite{LM14}. For example, when extracting a quantifier $\forall$ from some $\EX$ modality, we need to introduce two propositions, this can be seen for  the formula $\EX (\forall p. (\AX p \ou \AX \non p))$ which  is translated as~:
\[
 \exists z. \forall p. \forall z'. \Big(    (\EX (z \et z') \impl \AX(z \impl z')) \et  \EX (z \et   (\AX p \ou \AX \non p))\Big)
\]
where the proposition $z$ is used to mark a state, and $z'$ is used to  enforce that only at most one successor is labeled by $z$. 
Of course, these two remarks may have a strong impact on the complexity of the decision procedure. 
Finally, note also that the resulting \QBF\ formula is not in Prenex normal form.

\subsection{Variant of \MetFPF}

In the previous reduction, the modalities $\EU$ and $\AU$ may introduce an alternation of quantifiers~: 
an  atomic proposition $\kappa$ is introduced by an existential quantifier, and then an universal quantifier introduces a variable $z$ to encode the fixed point characterisation of $\Until$. We propose another reduction in order to avoid this alternation~: for this, we will use bit vectors (instead of single Boolean values) associated with every state to encode the distance from the current state to a state satisfying the right-hand side of the Until modality. 
 
First, we consider a $\QCTL$\ formula $\Phi$ under  negation normal form (NNF), where the negation is only applied to   atomic propositions. This transformation makes that $\Phi$ is built from temporal modalities in $S_{tmod} = \{\EX, \AX, \EU, \AU, \EW, \AW\}$

We can then reformulate Proposition~\ref{prop-flat} as follows~: 
\begin{proposition}
\label{prop-meth5}
Any  \QCTL\ formula $\Phi$ is equivalent to some \QCTL\ formula $\Psi$ in NNF of the form~:
${\displaystyle 
\Psi \eqdef  \calQ \:  \exists \{\kappa_1\ldots \kappa_m\}. \Big( \Phi_0 \et \ET_{i=1\ldots m} \AG (\kappa_i \impl   \theta_i)\Big)
}$
where $\calQ$ is a sequence of quantifications, $\Phi_0$ is a \CTL\ formula containing only the temporal modalities $\EX$, $\AX$ or $\AG$ and whose temporal height  is less or equal to 1, and  every  $\theta_i$ is a basic \CTL\ formula (with $1 \leq i \leq m$).  Moreover, $|\Psi|$ is in $O(|\Phi|)$. 
\end{proposition}
\begin{proof}
Consider w.l.o.g.\ a \QCTL\ formula   $\Phi$ in Prenex normal form and NNF. We can define $\Phi_0$ and the basic formulas   $\theta_i$s approximately   as in Proposition~\ref{prop-ctl-flat}, except that $\Phi_0$ contains only $\EX$,  $\AX$ or $\AG$ modalities, and every other modality gives rise to some quantified proposition $\kappa$ and a subformula $\AG(\ldots)$ in the main cunjunction of $\Psi$. Every $\theta_i$ starts with  
a modality in  $S_{tmod}$. Let $\vfi_i$ be the original $\Phi$-subformula   associated with $\theta_i$. 
Note that   $\Psi$ is in NNF, and  $\kappa_i$ occurs only once in $\Psi$ in the scope of a negation, and it happens in the subformula $\AG(\kappa_i \impl \theta_i)$. We now have to show that $\Phi$ is equivalent to $\Psi$. Consider the formula $\widetilde{\Psi}$ where every $\impl$ is replaced by $\equivaut$~: by following the same arguments of Proposition~\ref{prop-ctl-flat}, we clearly have $\Phi \equiv \widetilde{\Psi}$, and $\widetilde{\Psi} \impl \Psi$. It remains to prove the opposite direction.

To prove $\Psi \impl  \widetilde{\Psi}$,  it is sufficient to show that this is true for the empty $\calQ$ (as equivalence is substitutive). 
Assume $\calK, x \sat_\varepsilon \Psi$. Then there exists an environment $\varepsilon'$  from  $\{\kappa_1,\ldots,\kappa_m\}$ to $2^V$ such that $\calK, x \sat_{\varepsilon\compo\varepsilon'} \Phi_0 \et \ET_i \AG (\kappa_i \impl \theta_i)$. Now we have~:
\[
\forall i,\quad \calK, x  \sat_{\varepsilon \compo\varepsilon'} \theta_i \quad\impl\quad \calK, x  \sat_{\varepsilon'}  \vfi_i  \quad \mbox{and} \quad \calK, x  \sat_{\varepsilon \compo\varepsilon'}  \Phi_0 \quad\impl\quad \calK, x  \sat_{\varepsilon'}  \widetilde{\Phi_0}
\]
Indeed, assume that it is not true and $\calK, x  \sat_{\varepsilon \compo\varepsilon'} \theta_i$ and  $\calK, x  \not\sat_{\varepsilon'}  \vfi_i$.  Consider such a  formula $\vfi_i$ with the smallest temporal height. The only atomic propositions $\kappa_j$ occurring in $\theta_i$ are then associated with some $\theta_j$ and $\vfi_j$ which verify the property and thus any state satisfying such a $\theta_j$, also satisfies $\vfi_j$. Therefore any state labeled by such a $\kappa_j$ is correctly labeled (and satisfies $\vfi_j$). And the states that are not labeled by  $\kappa_j$ cannot make $\theta_i$ to be wrongly evaluated to true (because $\kappa_j$ is not in the scope of a negation). Therefore $\vfi_i$ holds true at $x$. The same holds for $\Phi_0$ and $\widetilde{\Phi_0}$.
As a direct consequence, we have  $\Psi \impl  \widetilde{\Psi}$.
\end{proof}
 
From the previous proposition, we  derive a new reduction to \QBF\ (called~\MetX). For modalities $\EW$ and $\AW$, we use the same encoding as for method \MetFP, except that we use the corresponding Boolean proposition $\kappa_i$ directly for the greatest fixed point. And for Until-based modalities, corresponding to least fixed points, we will use bit vectors instead of atomic propositions to encode the truth value of $\EU$ or $\AU$~: for  a formula  $\theta_i=\Ex \vfi \Until \psi$, we will consider a bit vector $\bar{\kappa_i}$  of length $\lceil\log(|V|+1)\rceil$ for every state instead of a single Boolean value  $\kappa_i$. The idea is that   in a state $x$, the value $\bar{\kappa_i}$ encodes in binary the distance (in terms of number of transitions) from $x$ to a state satisfying $\psi$ along a path satisfying  $\vfi$. And for $\theta_i=\All \vfi \Until \psi$, the value $\bar{\kappa_i}$ encodes the maximal distance before a state satisfying $\psi$ (along a path where $\vfi$ is true).
In the following, such a $\kappa_i$ associated to a $\theta_i$ based on an Until is called an Until-$\kappa$.
Note that given a bit vector $\bar{\kappa_i}$ and an integer value $d$ encoded in binary, we will use $[\bar{\kappa_i}^y=d]$ and $[\bar{\kappa_i}^y<d]$ to denote the corresponding  propositional formulas over $\bar{\kappa_i}^y$.

\begin{table}
\begin{align*}
\widehat{\vfi \et \psi}^{x,P} &\eqdef  \widehat{\vfi}^{x,P} \et \widehat{\psi}^{x,P}  \quad \quad
\widehat{\vfi \ou \psi}^{x,P} \eqdef  \widehat{\vfi}^{x,P} \ou \widehat{\psi}^{x,P} \quad \quad
\widehat{\non  p}^{x,P} \eqdef   \non \widehat{p}^{x,P} \\
\widehat{\EX  \vfi}^{x,P} & \eqdef   \OU_{(x,y) \in E} \widehat{\vfi}^{y,P} \quad \quad 
\widehat{\AX  \vfi}^{x,P} \eqdef   \ET_{(x,y) \in E} \widehat{\vfi}^{y,P} \quad \quad 
\widehat{\AG  \vfi}^{x,P} \eqdef   \ET_{(x,y) \in E^*} \widehat{\vfi}^{y,P}  \\
\text{For}\ Q\in\{\exists,\forall\} & \quad \widehat{Q p. \vfi}^{x,P} \eqdef 
Q p^{v_1} \ldots p^{v_n}  \widehat{\vfi}^{x,P\cup\{p\}}  \\
\text{For}\ Q\in\{\exists,\forall\} & \quad \widehat{Q \dot{\kappa}_i. \vfi}^{x,P} \eqdef \begin{cases} 
Q \bar{\kappa_i}^{v_1} \ldots \bar{\kappa_i}^{v_n}.  \widehat{\vfi}^{x,P\cup\{\bar{\kappa_i\}}} & \text{if} \ \dot{\kappa}_i \:\text{encodes a least fixed point.} \\  
Q \kappa_i^{v_1} \ldots \kappa_i^{v_n}.  \widehat{\vfi}^{x,P\cup\{\kappa_i\}} & \text{otherwise}  \end{cases} \\
\widehat{p}^{x,P} \eqdef & \begin{cases} p^x & \text{if}\: p \in P\ \text{and}\ p\ \text{is a Boolean var.}\\ 
[p^x<|V|] & \text{if}\: p \in P\ \text{and}\ p\ \text{is a Bit vect.}\\ 
\top & \text{if}\:  p \not\in P  \et p \in \ell(x) \\ 
\bot & \text{otherwise} \\ \end{cases} \\ 
\end{align*}
\caption{Transformation rules for method~\MetX}
\label{meth5}
\end{table}

The new reduction is  based on the rewriting rules of Table~\ref{meth5} for several operators, and 
we define the reduction for $\AG(\kappa_i \impl \theta_i)$ for $\theta_i=\Ex \vfi \WUntil \psi$  or $\theta_i=\All \vfi \WUntil \psi$  as follows~:

\begin{align}
\stackon[-8pt]{$\AG (\kappa_i \impl \Ex \vfi \WUntil \psi)$}{\vstretch{1.5}{\hstretch{15}{\widehat{\phantom{\;}}}}}^{\;\;x,P}  \eqdef &  \ET_{(x,y) \in E^*} \Big(  \kappa_i^y  \impl \big( \widehat{\psi}^{y,P} \ou (\widehat{\vfi}^{y,P} \et \OU_{(y,y')\in E} \kappa_i^{y'})\big)\Big)   \label{eq5EW} \\
\stackon[-8pt]{$\AG (\kappa_i \impl \All \vfi \WUntil \psi)$}{\vstretch{1.5}{\hstretch{15}{\widehat{\phantom{\;}}}}}^{\;\;x,P}  \eqdef &  \ET_{(x,y) \in E^*} \Big(  \kappa_i^y  \impl \big( \widehat{\psi}^{y,P} \ou (\widehat{\vfi}^{y,P} \et \ET_{(y,y')\in E} \kappa_i^{y'})\big)\Big)  
 \label{eq5AW}
\end{align}

And for $\AG(\kappa_i \impl \theta_i)$ with $\theta_i=\Ex \vfi \Until \psi$  or  $\theta_i=\All \vfi \Until \psi$, we have~:
\begin{align}
\stackon[-8pt]{$\AG (\bar{\kappa_i} \impl \Ex \vfi \Until \psi)$}{\vstretch{1.5}{\hstretch{15}{\widehat{\phantom{\;}}}}}^{\;\;x,P}  \eqdef &  \ET_{(x,y) \in E^*}   \Big[ \Big(  [\bar{\kappa_i}^y=0] \impl \widehat{\psi}^{y,P} \Big) \nonumber \\
 & \et \ET_{1\leq d <|V|} \Big( [\bar{\kappa_i}^y=d] \impl   \big(  \widehat{\vfi}^{y,P} \et \OU_{(y,y')\in E} [\bar{\kappa_i}^{y'}=d-1]\big)\Big) \Big]
 \label{eq5EU}
\end{align}
where $[\bar{\kappa_i}^y=d]$ and $[\bar{\kappa_i}^y<d]$  with a given value $d$ correspond to propositional formulas over $\bar{\kappa_i}^y$.
And for $\AG(\kappa_i \impl \theta_i)$ with $\theta_i=\All \vfi \Until \psi$, we have~:
\begin{align}
\stackon[-8pt]{$\AG (\bar{\kappa_i}\impl \All \vfi \Until \psi)$}{\vstretch{1.5}{\hstretch{15}{\widehat{\phantom{\;}}}}}^{\;\;x,P}  \eqdef &  \ET_{(x,y) \in E^*} \Big[ \Big(  [\bar{\kappa_i}^y=0] \impl \widehat{\psi}^{y,P} \Big) \nonumber \\
 & \et \ET_{1\leq d <|V|} \Big( [\bar{\kappa_i}^y=d] \impl   \big(  \widehat{\vfi}^{y,P} \et \ET_{(y,y')\in E} [\bar{\kappa_i}^{y'}<d]\big)\Big) \Big]
\label{eq5AU}
\end{align}

Note that in the equivalences \ref{eq5EU} and \ref{eq5AU}, we cannot replace the implication by an equivalence~: for a state $y$ s.t.\ $\widehat{\vfi}^{y,P}$, the existence of some successor $y'$ s.t.\ $[\bar{\kappa_i}^{y'}<d]$ is not sufficient to imply 
$[\bar{\kappa_i}^{y}=d]$. This is why we consider only implications here. 

Method~\MetX\  is defined from previous rules and we get a Prenex \QBF\  formula~:
\begin{corollary}
Given a \QCTL\ formula $\Phi$,  a Kripke structure $\calK=\tuple{V,E,\ell}$, and a state $x$, the Prenex \QBF\ formula $\widehat{\Flat(\Phi)}^{x}$ obtained by applying the rules of Tables~\ref{meth5} and rules~\ref{eq5EW}, \ref{eq5AW}, \ref{eq5EU},  and \ref{eq5AU} above, is valid iff $\calK, x \sat \Phi$, and $|\widehat{\Flat(\Phi)}^{x}|$ is in $O(|V|\cdot(|V|+|E|)\cdot |\Phi|)$. 
\end{corollary}

This approach allows us  to easily adapt the algorithm to bounded model-checking~: instead of considering values from $1$ to $|V|$ for $d$ in the definition of Until modalities, one can restrict the range to a smaller interval, to get a smaller \QBF\ formula to check. Note also that obtaining a Prenex \QBF\ formula is interesting in practice, because many \QBF-solvers require Prenex formulas as inputs.

We can also observe that this reduction proceeds a bit like the classical model-checking algorithm for \CTL\ where for deciding   $x \sat \Phi$, the truth value of every  $\Phi$-subformula is computed in every state of the model. In some case, this may induce additional work compared to methods like \MetFP\ or \MetUU\ (for example to decide whether the initial state satisfies a $\EU$ formula).  

\subsubsection{Using BitVectors for $\exists^1$ and  $\forall^1$.}

The quantifiers $\exists^1$ and $\forall^1$ are very useful in many specifications. It can be interesting to develop ad-hoc algorithms  in order to improve the generated \QBF\ formulas. For the first methods we described, they are translated into propositional formulas (instead of introducing extra quantified atomic propositions as they are formally defined). 
Another method consists in using bit vectors as in the treatment of Until modalities described above~: in this case, these quantifiers introduce a unique bit vector of size $\lceil\log(|V|+1)\rceil$ to store the number of the state selected by the quantifier (that is the state  which will be labeled by the proposition). This method is interesting, since it reduces the number of quantified propositions, it is integrated in the method~\MetX.


\section{Experimental results}

\label{sec-experiments}

In this section, we consider three examples to illustrate the \QBF-based model-checking approach for \QCTL. We propose these problems because  the size of the   structures can easily be scaled up, and the  properties to be checked are quite different. Note that these properties cannot be expressed  with classical temporal logics.

We have implemented a prototype to try the different reduction strategies~\footnote{NB: For reductions~\MetFPF\ and~\MetX\ the formula has to be given in Prenex normal form.} . Our tool is available online~\footnote{\url{https://www.irif.fr/~francoisl/qctlmc.html}}~: given a  Kripke structure $\calK$ with a state $x$ and a formula $\Phi$, it produces a specification file (corresponding to $\widehat{\Phi}^{x}$) for the SMT-solver Z3~\cite{MouraB08}. The choice of Z3 was motivated by the fact that the generated \QBF\ formulas are not always Prenex, which many \QBF-solvers require, unlike Z3.

\subsection{$k$-connectivity}

Here, we consider an undirected graph, and we want to check whether there exist (at least) $k$ internally disjoint paths~\footnote{Two paths  paths $src \leftrightarrow r_1 \leftrightarrow \ldots \leftrightarrow r_k \leftrightarrow dest$  and $src \leftrightarrow r'_1 \leftrightarrow \ldots \leftrightarrow r'_{k'} \leftrightarrow dest$  are internally disjoint iff $r_i  \not= r'_{j}$ for any $1 \leq i \leq k$ and $1 \leq j \leq k'$. And note that with this def., if there is an edge $(x,y)$, there exist $k$ internally disjoint paths from $x$ to $y$ for any $k$.} from a vertex $x$ to some vertex $y$.
A classical result in graph theory due to Menger   ensures that, given two vertices $x$ and $y$ in a graph $G$, the minimum number of vertices whose deletion makes that there is no more paths between $x$ and $y$  is equal to the maximum number of internally disjoint paths between these two vertices.

We can encode these two ideas by the following \QCTL\ formulas (interpreted over $x$)~:
\begin{align}
\Phi_k \;&\eqdef\; \exists p_1 \ldots \exists p_{k-1} \:\Big( \ET_{1\leq i < k} \EX \big(\Ex (p_i \et \ET_{j\not= i} p_j) \:\Until\: y\big) \;\et\; \EX\: \Ex (\ET_{1\leq i <k} \non p_i) \:\Until\: y \Big)  \label{kconF1}\\
\Psi_k \;&\eqdef\; \forall^1 p_1 \ldots \forall^1 p_{k-1} \: \EX  \: \Big( \Ex \big( \ET_{1\leq i < k} \non p_i\big) \:\Until\: y\Big) \label{kconF2}
\end{align}

$\Phi_k$ uses the labelling by the $p_i$'s to mark the internal vertices of $k$ paths between the current position and the vertex $y$. The modality $\EX$ is used to consider only the intermediate states (and not the starting state). The formula $\Psi_k$ proceeds differently~: the idea is to mark exactly $k-1$ states with $p_1,\ldots,p_{k-1}$ and to verify that there still exists at least one path leading to $y$ without going through the states labeled by some $p_i$. 

By Menger's Theorem, we know that these formulas are equivalent over undirected graphs.

We interpret these formulas over Kripke structures $\calS_{n,m}$ with $n\geq m$ (see Figure~\ref{fig-kconnect}) which correspond to two kinds of grid $n\times n$ connected by $m$ edges (these edges are of the form $(q_{i,n},r_{1,i})$ or $(q_{n,i},r_{i,1})$
The initial state is $q_{0,0}$ and when evaluating $\Phi_k$ or $\Psi_k$ we assume the state $r_{n,n}$ to be labeled by $y$.  In this context, we clearly have that $\Phi_k$ and $\Psi_k$ hold for true at $q_{0,0}$ iff $k\leq m$.

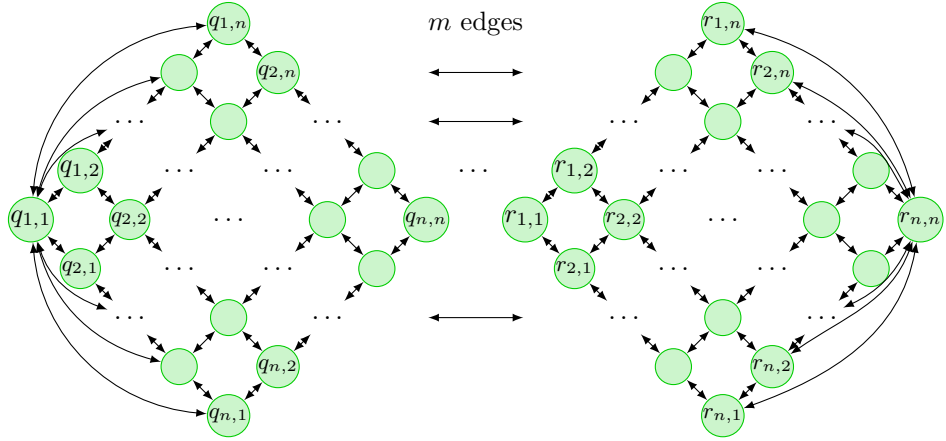
\begin{figure}[h]
\centering
\begin{tikzpicture}[inner sep=0pt,scale=1.3]
\draw (0,2) node[minirond,vert] (q00) {$q_{1,1}$} ;
\draw (0.5,2.5) node[minirond,vert] (q01) { $q_{1,2}$} ;
\draw (1,3) node[minirond,white] (q02) {\textcolor{black}{$\;\ldots\;$}} ;
\draw (1.5,3.5) node[minirond,vert] (q03) {\footnotesize {$\quad\;\;$}} ;
\draw (2,4) node[minirond,vert] (q04) {\footnotesize $q_{1,n}$} ;
\draw (0.5,1.5) node[minirond,vert] (q10) {\footnotesize $q_{2,1}$} ;
\draw (1,2) node[minirond,vert] (q11) {\footnotesize $q_{2,2}$} ;
\draw (1.5,2.5) node[minirond,white] (q12) {\textcolor{black}{$\;\ldots\;$}} ;
\draw (2,3) node[minirond,vert] (q13) {\footnotesize {$\quad\;\;$}} ;
\draw (2.5,3.5) node[minirond,vert] (q14) {\footnotesize $q_{2,n}$} ;
\draw (1,1) node[minirond,white] (q20) {\textcolor{black}{$\;\ldots\;$}} ;
\draw (1.5,1.5) node[minirond,white] (q21) {\textcolor{black}{$\;\ldots\;$}} ;
\draw (2,2) node (q22) { $\;\ldots\;$} ;
\draw (2.5,2.5) node[minirond,white] (q23) {\textcolor{black}{$\;\ldots\;$}} ;
\draw (3,3) node[minirond,white] (q24) {\textcolor{black}{$\;\ldots\;$}} ;
\draw (1.5,0.5) node[minirond,vert] (q30) {\footnotesize $\quad\;\;$} ;
\draw (2,1) node[minirond,vert] (q31) {\footnotesize $\quad\;\;$} ;
\draw (2.5,1.5) node[minirond,white] (q32) {\textcolor{black}{$\;\ldots\;$}} ;
\draw (3,2) node[minirond,vert] (q33) {\footnotesize {$\quad\;\;$}} ;
\draw (3.5,2.5) node[minirond,vert] (q34) {\footnotesize $\quad\;\;$} ;
\draw (2,0) node[minirond,vert] (q40) {\footnotesize $q_{n,1}$} ;
\draw (2.5,.5) node[minirond,vert] (q41) {\footnotesize $q_{n,2}$} ;
\draw (3,1) node[minirond,white] (q42) {\textcolor{black}{$\;\ldots\;$}} ;
\draw (3.5,1.5) node[minirond,vert] (q43) {\footnotesize {$\quad\;\;$}} ;
\draw (4,2) node[minirond,vert] (q44) {\footnotesize $q_{n,n}$} ;

\draw[>=latex,<->](q00) -- (q01) ;
\draw[>=latex,<->](q01) -- (q02) ;
\draw[>=latex,<->](q02) -- (q03) ;
\draw[>=latex,<->](q03) -- (q04) ;
\draw[>=latex,<->](q10) -- (q11) ;
\draw[>=latex,<->](q11) -- (q12) ;
\draw[>=latex,<->](q12) -- (q13) ;
\draw[>=latex,<->](q13) -- (q14) ;
\draw[>=latex,<->](q30) -- (q31) ;
\draw[>=latex,<->](q31) -- (q32) ;
\draw[>=latex,<->](q32) -- (q33) ;
\draw[>=latex,<->](q33) -- (q34) ;
\draw[>=latex,<->](q40) -- (q41) ;
\draw[>=latex,<->](q41) -- (q42) ;
\draw[>=latex,<->](q42) -- (q43) ;
\draw[>=latex,<->](q43) -- (q44) ;
\draw[>=latex,<->](q00) -- (q10) ;
\draw[>=latex,<->](q10) -- (q20) ;
\draw[>=latex,<->](q20) -- (q30) ;
\draw[>=latex,<->](q30) -- (q40) ;
\draw[>=latex,<->](q01) -- (q11) ;
\draw[>=latex,<->](q11) -- (q21) ;
\draw[>=latex,<->](q21) -- (q31) ;
\draw[>=latex,<->](q31) -- (q41) ;
\draw[>=latex,<->](q03) -- (q13) ;
\draw[>=latex,<->](q13) -- (q23) ;
\draw[>=latex,<->](q23) -- (q33) ;
\draw[>=latex,<->](q33) -- (q43) ;
\draw[>=latex,<->](q04) -- (q14) ;
\draw[>=latex,<->](q14) -- (q24) ;
\draw[>=latex,<->](q24) -- (q34) ;
\draw[>=latex,<->](q34) -- (q44) ;
\draw[>=latex,<->,out=73,in=-160] (q00) to  (q02);
\draw[>=latex,<->,out=75,in=-165] (q00) to  (q03);
\draw[>=latex,<->,out=85,in=-175] (q00) to  (q04);

\draw[>=latex,<->,out=-73,in=160] (q00) to  (q20);
\draw[>=latex,<->,out=-75,in=165] (q00) to  (q30);
\draw[>=latex,<->,out=-85,in=175] (q00) to  (q40);

\begin{scope}[shift={(5,0)}]
\draw (0,2) node[minirond,vert] (r00) {$r_{1,1}$} ;
\draw (0.5,2.5) node[minirond,vert] (r01) { $r_{1,2}$} ;
\draw (1,3) node[minirond,white] (r02) {\textcolor{black}{$\;\ldots\;$}} ;
\draw (1.5,3.5) node[minirond,vert] (r03) {\footnotesize {$\quad\;\;$}} ;
\draw (2,4) node[minirond,vert] (r04) {\footnotesize $r_{1,n}$} ;
\draw (0.5,1.5) node[minirond,vert] (r10) {\footnotesize $r_{2,1}$} ;
\draw (1,2) node[minirond,vert] (r11) {\footnotesize $r_{2,2}$} ;
\draw (1.5,2.5) node[minirond,white] (r12) {\textcolor{black}{$\;\ldots\;$}} ;
\draw (2,3) node[minirond,vert] (r13) {\footnotesize {$\quad\;\;$}} ;
\draw (2.5,3.5) node[minirond,vert] (r14) {\footnotesize $r_{2,n}$} ;
\draw (1,1) node[minirond,white] (r20) {\textcolor{black}{$\;\ldots\;$}} ;
\draw (1.5,1.5) node[minirond,white] (r21) {\textcolor{black}{$\;\ldots\;$}} ;
\draw (2,2) node (r22) { $\;\ldots\;$} ;
\draw (2.5,2.5) node[minirond,white] (r23) {\textcolor{black}{$\;\ldots\;$}} ;
\draw (3,3) node[minirond,white] (r24) {\textcolor{black}{$\;\ldots\;$}} ;
\draw (1.5,0.5) node[minirond,vert] (r30) {\footnotesize $\quad\;\;$} ;
\draw (2,1) node[minirond,vert] (r31) {\footnotesize $\quad\;\;$} ;
\draw (2.5,1.5) node[minirond,white] (r32) {\textcolor{black}{$\;\ldots\;$}} ;
\draw (3,2) node[minirond,vert] (r33) {\footnotesize {$\quad\;\;$}} ;
\draw (3.5,2.5) node[minirond,vert] (r34) {\footnotesize $\quad\;\;$} ;
\draw (2,0) node[minirond,vert] (r40) {\footnotesize $r_{n,1}$} ;
\draw (2.5,.5) node[minirond,vert] (r41) {\footnotesize $r_{n,2}$} ;
\draw (3,1) node[minirond,white] (r42) {\textcolor{black}{$\;\ldots\;$}} ;
\draw (3.5,1.5) node[minirond,vert] (r43) {\footnotesize {$\quad\;\;$}} ;
\draw (4,2) node[minirond,vert] (r44) {\footnotesize $r_{n,n}$} ;

\draw[>=latex,<->](r00) -- (r01) ;
\draw[>=latex,<->](r01) -- (r02) ;
\draw[>=latex,<->](r02) -- (r03) ;
\draw[>=latex,<->](r03) -- (r04) ;
\draw[>=latex,<->](r10) -- (r11) ;
\draw[>=latex,<->](r11) -- (r12) ;
\draw[>=latex,<->](r12) -- (r13) ;
\draw[>=latex,<->](r13) -- (r14) ;
\draw[>=latex,<->](r30) -- (r31) ;
\draw[>=latex,<->](r31) -- (r32) ;
\draw[>=latex,<->](r32) -- (r33) ;
\draw[>=latex,<->](r33) -- (r34) ;
\draw[>=latex,<->](r40) -- (r41) ;
\draw[>=latex,<->](r41) -- (r42) ;
\draw[>=latex,<->](r42) -- (r43) ;
\draw[>=latex,<->](r43) -- (r44) ;
\draw[>=latex,<->](r00) -- (r10) ;
\draw[>=latex,<->](r10) -- (r20) ;
\draw[>=latex,<->](r20) -- (r30) ;
\draw[>=latex,<->](r30) -- (r40) ;
\draw[>=latex,<->](r01) -- (r11) ;
\draw[>=latex,<->](r11) -- (r21) ;
\draw[>=latex,<->](r21) -- (r31) ;
\draw[>=latex,<->](r31) -- (r41) ;
\draw[>=latex,<->](r03) -- (r13) ;
\draw[>=latex,<->](r13) -- (r23) ;
\draw[>=latex,<->](r23) -- (r33) ;
\draw[>=latex,<->](r33) -- (r43) ;
\draw[>=latex,<->](r04) -- (r14) ;
\draw[>=latex,<->](r14) -- (r24) ;
\draw[>=latex,<->](r24) -- (r34) ;
\draw[>=latex,<->](r34) -- (r44) ;
\draw[>=latex,<->,out=-15,in=105] (r04) to  (r44);
\draw[>=latex,<->,out=-25,in=115] (r14) to  (r44);
\draw[>=latex,<->,out=-20,in=120] (r24) to  (r44);
\draw[>=latex,<->,out=20,in=-100] (r40) to  (r44);
\draw[>=latex,<->,out=35,in=-110] (r41) to  (r44);
\draw[>=latex,<->,out=25,in=-115] (r42) to  (r44);
\end{scope}

\draw (4,3.5) node (aux1) {} ;
\draw (5,3.5) node (aux1p) {} ;
\draw (4,3) node (aux2) {} ;
\draw (5,3) node (aux2p) {} ;
\draw (4,1) node (aux3) {} ;
\draw (5,1) node (aux3p) {} ;
\draw[>=latex,<->](aux1) -- (aux1p) ;
\draw[>=latex,<->](aux2) -- (aux2p) ;
\draw[>=latex,<->](aux3) -- (aux3p) ;
\draw (4.5,2.5) node (aux1) {\ldots} ;

\draw (4.5,4) node (aux1) {$m$ edges} ;

\end{tikzpicture}
\caption{Structure $\calS_{n,m}$ for the $k$-connectivity problem.}
\label{fig-kconnect}
\end{figure}

Detailed results are presented in Appendix~\ref{app-res}. The main lessons we can see are~:
\begin{itemize}
\item Formula $\Phi_k$ is much more difficult to verify~: the number of temporal modalities is probably one explanation. Another one for methods \MetFP\ and \MetFPF (based on the fixed point characterization of $\Until/\WUntil$ modalities) could be an alternation of quantifiers~: in $\Phi_k$, there is an existential quantification over the $p_i$s and the $\EU$s introduce an universal quantification, but it is not the case for $\Psi_k$, where there are only universal quantifications for these reductions.
\item The reduction \MetFP\ is the most efficient~: it can be used to verify models with more than two thousands states when $m$ is small. The method \MetFPF\ is also rather efficient, but the flattening seems to be too costly for such a simple formula.
\item The reduction \MetUU\ produces very large \QBF\ formulas, but rather simple to check.
\end{itemize}

We can generalise the problem by verifying that there exist at least $k$ internally disjoint paths between any pair of reachable vertices $x$ and $y$ in a given structure. These previous formulas can be modified  as follows~: 
\begin{align}
\Phi^g_k \;&\eqdef\; \forall^1 y \exists p_1 \ldots \exists p_{k-1} \: \AG \Big( \!\!\ET_{1\leq i < k} \!\!\EX \big(\Ex (p_i \et \ET_{j\not= i} p_j) \:\Until\: y\big) \;\et\; \EX\: \Ex (\!\!\ET_{1\leq i <k} \!\!\non p_i) \:\Until\: y \Big) \label{kconGF1} \\
\Psi^g_k \;&\eqdef\; \forall^1 y \forall^1 p_1 \ldots \forall^1 p_{k-1}  \AG \Big[ \: \EX  \: \Big( \Ex \big( \ET_{1\leq i < k} \non p_i\big) \:\Until\: y\Big)\Big] \label{kconGF2} 
\end{align}

In that case, $\Phi^g_k$ is useless~: too complex to be verified. And the method \MetFP\ is still the most efficient.

\subsection{Nim game}

Nim game is a turn-based two-players game. A configuration is a set of heaps of objects and a boolean value indicating whose turn it is. At each turn, a player has to choose one non-empty heap and remove at least one object from it. The aim of each player is to remove the last object. Given a configuration $c$ and a Player-$J$ with $J\in\{1,2\}$, we can build a finite Kripke structure $\calS_J$, where $x_c$ is a state corresponding to the configuration $c$ ; and  use a \QCTL\ formula $\Phiwin^J$ such that $\calS, x_c \sat \Phiwin^J$ iff Player-$J$ has a wining strategy from $c$. Note that there is a simple and well-known criterion over the numbers of objects in each heap to decide who has a  winning strategy, but we consider this problem just because it is interesting to illustrate what kind of problem we can solve with  \QCTL. 

Each  configuration corresponds to a state in $\calS_J$. Every move for Player-$\bar{J}$ from a configuration $c$ to a configuration $c'$ provides a transition $(x_c,x_{c'})$ in $\calS_J$. However, a move of Player-$J$ from  $c$ to $c'$ is encoded as two transitions $x_c \fleche x_{c,c'} \fleche x_{c'}$ where $x_{c,c'}$ is then an intermediary state we   use to encode a strategy for Player-$J$ (marking $x_{c,c'}$ by an atomic proposition will correspond to Player-$J$ choosing $c'$ from  $c$). We assume that every state $x_c$ is labeled by $t_1$  if it's Player-$1$'s turn to play at $c$, and by $t_2$ otherwise. Every intermediary state $x_{c,c'}$ is labeled by $\text{int}$. 
We also label empty configurations by $w_1$ or $w_2$, depending on which player played the last move. 

Clearly, the size of $\calS$ will depend on the number of objects in each set in the initial configuration. The formula $\Phiwin^J$ depends only on $J$~:
\[
\Phiwin^J \;\eqdef\; \exists m. \Big( \AG \big(t_J \:\impl \: \EX m \big) \;\et\; \AF \big(w_J \ou (\text{int} \et \non m)\big)\Big)
\]
This formula holds true in a state corresponding to some configuration $c$ iff there  exists a labelling by $m$ such that every reachable configuration where it's Player-$J$'s turn, has a successor labeled by $m$ (thus a possible choice to do)  and every execution from the current state leads to either a winning state for Player-$J$ or a non-selected intermediary state, therefore all outcomes induced by the underlying strategy have to  verify $\F w_J$.  Note that in this example, the Kripke structure is acyclic (except the self-loops on the ending states).


From detailed results in appendix, we can see~:
\begin{itemize}
\item One can consider structures with more than 10 thousand states. Note that the number of heaps is important for the size of the model, but the maximal length of a game depends on the number of objects (and is rather small in our examples). 
\item The most efficient method is \MetFPF\ (with \MetFP).
\item Method \MetX\ is more efficient  when we  consider bounded model-checking. Note that in this case, the verification may not be complete~: if the \QBF\ formula is valid, the property is satisfied by the structure, otherwise, no conclusion can be done, except if we can prove that the chosen bound was big enough to be sure that there is no solution. In our case, we can easily compute the maximal bound~: at each turn, a player has to pick at least one object, such a move may give rise to one transition in the model (for the opponent), or two transitions (for the player for whom we look for a strategy). Thus, if there are $n$ objects in the initial configuration, we can choose $\frac{3n}{2}$ for the bound.
\end{itemize}

\subsection{Resources distribution}

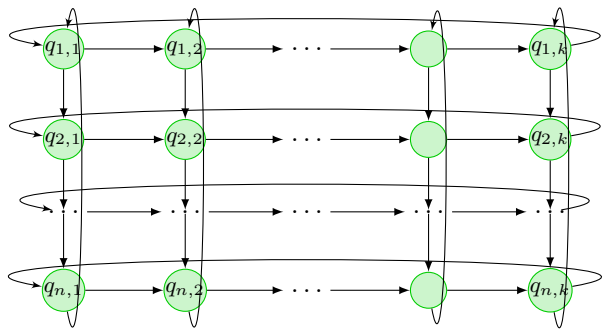
\begin{figure}[t]
\centering
\begin{tikzpicture}[inner sep=0pt,scale=0.8]
\draw (0,4) node[minirond,vert] (v00) {\footnotesize $q_{1,1}$} ;
\draw (2,4) node[minirond,vert] (v01) {\footnotesize $q_{1,2}$} ;
\draw (4,4) node (v02) { $\;\ldots\;$} ;
\draw (6,4) node[minirond,vert] (v03) {\footnotesize {$\quad\;\;$}} ;
\draw (8,4) node[minirond,vert] (v04) {\footnotesize $q_{1,k}$} ;
\draw (0,2.5) node[minirond,vert] (v10) {\footnotesize $q_{2,1}$} ;
\draw (2,2.5) node[minirond,vert] (v11) {\footnotesize $q_{2,2}$} ;
\draw (4,2.5) node  (v12) { $\;\ldots\;$} ;
\draw (6,2.5) node[minirond,vert] (v13) {\footnotesize $\quad\;\;$} ;
\draw (8,2.5) node[minirond,vert] (v14) {\footnotesize $q_{2,k}$} ;
\draw (0,0) node[minirond,vert] (v30) {\footnotesize $q_{n,1}$} ;
\draw (2,0) node[minirond,vert] (v31) {\footnotesize $q_{n,2}$} ;
\draw (4,0) node (v32) { $\;\ldots\;$} ;
\draw (6,0) node[minirond,vert] (v33) {\footnotesize $\quad\;\;$} ;
\draw (8,0) node[minirond,vert] (v34) {\footnotesize $q_{n,k}$} ;
\draw (0,1.3) node  (v20) { $\;\ldots\;$} ;
\draw (2,1.3) node  (v21) { $\;\ldots\;$} ;
\draw (4,1.3) node  (v22) { $\;\ldots\;$} ;
\draw (6,1.3) node  (v23) { $\;\ldots\;$} ;
\draw (8,1.3) node  (v24) { $\;\ldots\;$} ;

\draw[-latex](v00) -- (v01) ;
\draw[-latex](v01) -- (v02) ;
\draw[-latex](v02) -- (v03) ;
\draw[-latex](v03) -- (v04) ;
\draw[-latex](v10) -- (v11) ;
\draw[-latex](v11) -- (v12) ;
\draw[-latex](v12) -- (v13) ;
\draw[-latex](v13) -- (v14) ;
\draw[-latex](v20) -- (v21) ;
\draw[-latex](v21) -- (v22) ;
\draw[-latex](v22) -- (v23) ;
\draw[-latex](v23) -- (v24) ;
\draw[-latex](v30) -- (v31) ;
\draw[-latex](v31) -- (v32) ;
\draw[-latex](v32) -- (v33) ;
\draw[-latex](v33) -- (v34) ;

\draw[-latex](v00) -- (v10) ;
\draw[-latex](v01) -- (v11) ;
\draw[-latex](v03) -- (v13) ;
\draw[-latex](v04) -- (v14) ;
\draw[-latex](v10) -- (v20) ;
\draw[-latex](v11) -- (v21) ;
\draw[-latex](v13) -- (v23) ;
\draw[-latex](v14) -- (v24) ;
\draw[-latex](v20) -- (v30) ;
\draw[-latex](v21) -- (v31) ;
\draw[-latex](v23) -- (v33) ;
\draw[-latex](v24) -- (v34) ;

\draw[-latex',out=10,in=170] (v04) to  (v00);
\draw[-latex',out=10,in=170] (v14) to  (v10);
\draw[-latex',out=10,in=170] (v24) to  (v20);
\draw[-latex',out=10,in=170] (v34) to  (v30);

\draw[-latex',out=-80,in=80] (v30) to  (v00);
\draw[-latex',out=-80,in=80] (v31) to  (v01);
\draw[-latex',out=-80,in=80] (v33) to  (v03);
\draw[-latex',out=-80,in=80] (v34) to  (v04);

\end{tikzpicture}
\vspace{-0.2cm}
\caption{Structure $\calK_{n,k}$ for the resources distribution problem.}
\label{fig-resources}
\end{figure}

The last example is as follows~: given a Kripke structure $\calS$ and two integers $k$ and $d$, we aim at choosing at most $k$ states (called targets in the following) such that every reachable state (from the initial one) can reach a target in less than $d$ transitions. This problem can be encoded with the following \QCTL\ formula where $d$ modalities $\EX$ are nested~:
\[
\Phires \;\eqdef\;  \exists^1  c_1 \ldots \exists^1 c_k \: \AG \Big( (\!\!\OU_{1\leq i \leq k}\!\! c_i) \ou \EX \Big( 
(\!\!\OU_{1\leq i \leq k}\!\! c_i) \ou \Big( \ldots \ou \EX ( \!\!\OU_{1\leq i \leq k}\!\! c_i) \Big) \Big) \Big)
\]

For experimental  results, we consider the grid $\calK_{n,m}$ described at Figure~\ref{fig-resources}. Note that for this example, reductions~\MetUU\ and \MetFP\ are similar because there is no Until in the formula ($\AG$ is treated as a  conjunction). 
From detailed results in appendix, one can see~:
\begin{itemize}
\item Only small models have been successfully verified.
\item The nesting of $\EX$s operators give an advantage to the reduction based on flattening (\MetX\ and \MetFPF)~: the size of the \QBF\ formula increases more slowly.
\item The reduction \MetX\ is the most efficient on this example. Since there is no Until modality (except behind $\AG$ which is treated separately), the difference with \MetFPF\ is due to the encoding of $\exists_1$ operator with a unique bit vector in \MetX, this choice seems to be more efficient in this example. 
\end{itemize}


\section{Conclusion}

We have presented several reductions from \QCTL\ model-checking to \QBF. This provides a first tool for \QCTL\ model-checking. Of course, this is an ongoing-work, and many improvements, are possible~: the reduction strategies are still naive and could be significantly improved, and a better understanding of \QBF-solver would also be helpful to produce more efficient formulas (we have not yet tried to normalise formulas in a specific form which is often a crucial aspect in \SAT/\QBF-solving). Still, these first results are rather interesting and encouraging. They show the importance of writing "good"  \QCTL\ formulas for which the solver will be able to provide a result (this problem already exists for classical temporal logics, but it is more significant here due to the complexity induced by the quantifications). The examples also show that there is no "one best strategy": it  depends on the structure of the considered formula, and then offering several reduction strategies seems to be necessary in a \QBF-based model-checker for \QCTL.
Finally this work is also interesting because it could easily be adapted for other logics (like Sabotage logics~\cite{Benthem05}). In the future, we plan to continue to work on reduction strategies, and to use other \QBF-solvers.



\bibliography{biblio}

\appendix 


\section{Proof of Theorem~\ref{theo-UU}}

\label{app-uu}

\begin{proof}
We assume $\Phi$ to be fixed and we prove the property by structural induction over $\vfi$. Boolean operators are omitted. 
\begin{itemize}
\item $\vfi = p$~: if  $\calK,x\sat_\varepsilon p$, then either $p$ is a quantified proposition and  it belongs to $\dom(\varepsilon)$ and $x$ belongs to $\epsilon(p)$ and thus $v_\varepsilon \sat p^x$ by def.\ of $v_\varepsilon$, or $p$ belongs to $\ell(x)$. In  both cases we have $v_\varepsilon \sat \widehat{p}^{x,\dom(\varepsilon)}$. The converse is similar. 
\item $\vfi=\EX \psi$~:   $\calK,x\sat_\varepsilon \EX \psi$ iff there exists $(x,x')\in E$ s.t.\ $\calK,x ' \sat_\varepsilon \psi$, iff (by i.h.)  there exists $(x,x')\in E$ s.t.\  $v_\varepsilon \sat \widehat{\psi}^{x',\dom(\varepsilon)}$ which is equivalent to $v_\varepsilon \sat \widehat{\EX \psi}^{x,\dom(\varepsilon)}$. 
\item $\vfi=\exists p. \psi$. We have $\calK,x\sat_\varepsilon \exists p. \psi$ iff there exists  $V'\subseteq V$ s.t.\ $\calK,x \sat_{\varepsilon[p\mapsto V']} \psi$, iff (i.h.) $v_{\varepsilon[p\mapsto V']} \sat \widehat{\psi}^{x,\dom(\varepsilon)\cup\{p\}}$ which is equivalent to $v_{\varepsilon} \sat \exists p^{v_1}\ldots p^{v_n}.  \widehat{\psi}^{x,\dom(\varepsilon)\cup\{p\}}$ (by def.\ of    $v_\varepsilon$). 
\item $\vfi=\Ex \psi_1 \Until \psi_2$~: The definition of $\widehat{\vfi}^{x,\dom(\varepsilon)}$ corresponds to a finite unfolding of the expansion law that characterizes the $\EU$\ modality. Assume  $\calK,x\sat_\varepsilon \vfi$. There exists a  path $\rho \in 
\Path^\omega_\calK(x)$ and a position $i\geq 0$ s.t.\ $\rho(i) \sat_\varepsilon  \psi_2$ and $\rho(k) \sat_\varepsilon  \psi_1$ for any $0 \leq k < i$. The finite prefix $x=\rho(0)\cdots\rho(k)$ can be assumed to be simple, and then $k<|V|$. By using i.h., we get 
$v_\varepsilon \sat \widehat{\psi_2}^{\rho(i),\dom(\varepsilon)}$, and $v_\varepsilon \sat \widehat{\psi_1}^{\rho(k),\dom(\varepsilon)}$ for any $0\leq k <i$. From this point, the reader can easily verify by induction (starting at $i$, down to 0) that $v_\varepsilon\sat\overline{\Ex \psi_1\Until\psi_2}^{\rho(k), dom(\varepsilon), \{\rho(j)|j\leq k\}}$ for all $k\leq i$. This makes $\widehat{\vfi}^{x,\dom(\varepsilon)}$ to be satisfied by $v_\varepsilon$.

Conversely, assume $v_\varepsilon \sat \widehat{\vfi}^{x,\dom(\varepsilon)}$. Let us build a finite path $x=x_0, x_1, \cdots x_i$ that satisfies $\psi_1\Until\psi_2$. The first vertex is of course $x$, and we have $v_\varepsilon\sat\overline{\Ex\psi_1\Until\psi_2}^{x, dom(\varepsilon), \{x\}}$. Fix $i$ so that for every $0\leq k\leq i$, $x_k$ is built, $v_\varepsilon\sat\overline{\Ex\psi_1\Until\psi_2}^{x_k, dom(\varepsilon), \{x_j|j\leq k\}}$, $v_\varepsilon\not\sat\widehat{\psi_2}^{x_k, dom(\varepsilon)}$, and $(x_k, x_{k+1})\in E$ when $k<i$. Then, there must be $(x_i, y)\in E$ so that $y\not\in\{x_j|j\leq i\}$, and $v_\varepsilon\sat\overline{\Ex\psi_1\Until\psi_2}^{y, dom(\varepsilon), \{x_j|j\leq i\}\cup\{y\}}$, so we set $x_{i+1}=y$. The sequence eventually stops, because $V$ is finite and the path is simple. If $x_i$ is its last vertex, then we must have $v_\varepsilon\sat\widehat{\psi_2}^{x_i, dom(\varepsilon)}$, and $v_\varepsilon\sat\widehat{\psi_1}^{x_k, dom(\varepsilon)}$ for every $0\leq k<i$. By using i.h., we get $\calK, x\models_\varepsilon\phi$.
\item $\vfi=\All \psi_1 \Until \psi_2$~: this case is similar to the previous one, except that we have to consider loops.
Assume $\calK,x\sat_\varepsilon \vfi$. Then any path issued from $x$  satisfies $\psi_1 \Until \psi_2$. If $x$ contains a self-loop, then $\psi_2$ has to be satisfied  at $x$ (this is ensured by the first case in the def.\ of $\widehat{\vfi}^{x,\dom(\varepsilon)}$). Otherwise we consider all the paths from $x$~: either there is a simple prefix witnessing $\psi_1 \Until \psi_2$, or there is a loop from some point. In the   latter case, one of the state in the loop has to verify $\psi_2$. In both cases, the definition of   $\widehat{\vfi}^{x,\dom(\varepsilon)}$ gives the result. 
\end{itemize}
\end{proof}

\section{Experimental results} 
\label{app-res}

Detailed results for the three examples are presented in this section. In every case, we distinguished the time required to build the \QBF\ formula (the Z3 specification) and the time required to solve it. Times are given in seconds.

\subsection{$k$-connectivity}

\begin{small}
\[
\begin{array}{|l|c|c|c|c|c|}
\hline
 n, m & 3, 2 & 4, 3 & 4, 3 & 5, 4 & 35, 4 \\
 \text{formula} & \Psi_2 & \Psi_4 & \Phi_4 & \Psi_5 & \Psi_4\\
 \hline 
 \text{\# states} & 18 & 32 & 32 & 50 & 2450 \\
 \text{res} & \text{sat} & \text{unsat} & \text{unsat} & \text{unsat} & \text{sat}\\ 
 \hline
\hline
\multicolumn{6}{|c|}{\text{Reduction}\:  \MetUU} \\ 
\hline
\text{Time to build z3 form.\ } & 0.01 &  38.45 & 70.15 & - & - \\
\text{Size of z3 form.\ } & 7682 & 7 922 175 & 29 983 034 & - & - \\
\text{Time to solve z3 form.\ }&  0.04 & 2.38 & 8.61 & -  & - \\
\hline
\hline
\multicolumn{6}{|c|}{\text{Reduction} \: \MetFP} \\ 
\hline
\text{Time to build z3 form.\ } & 0 & 0.03 & 0.07 & 0.09  & 388 \\
\text{Size of z3 form.\ } & 1313 & 9097 & 11234 & 27 416 & 38 675 099 \\
\text{Time to solve z3 form.\ }& 0.03 & 0.06 & - & 0.16 & 132.67 \\
\hline
\hline
\multicolumn{6}{|c|}{\text{Reduction}\:  \MetFPF} \\ 
\hline
\text{Time to build z3 form.\ } & 0.03 & 0.1 & 0.25 & 0.1 & - \\
\text{Size of z3 form.\ } &  3666 & 20 944 & 56196 & 61 519 & - \\
\text{Time to solve z3 form.\ }& 0.12 & 0.28 & - & 57.96 & - \\
\hline
\hline
\multicolumn{6}{|c|}{\text{Reduction}\: \MetX} \\ 
\hline
\text{Time to build z3 form.\ } & 0.01 & 1 & 0.23 & 1.3 & - \\
\text{Size of z3 form.\ } & 6363 & 31243 & 107120 & 87681 & - \\
\text{Time to solve z3 form.\ }& 40.16 &  - & 43.09 & - & - \\
\hline
\hline
\end{array}
\]
\end{small}

\begin{small}
\[
\begin{array}{|l|c|c|c|c|}
\hline
 n, m & 3, 1 & 3, 3 & 9, 2 & 27, 2 \\
 \text{formula} & \Psi_2^g & \Psi_3^g & \Psi_3^g & \Psi_3^g \\
 \hline 
 \text{\# states} & 18 & 18 & 162 & 1458 \\
 \text{res} & \text{unsat} & \text{sat} & \text{unsat} & \text{unsat} \\ 
 \hline
\hline
\multicolumn{5}{|c|}{\text{Reduction}\:  \MetUU} \\ 
\hline
\text{Time to build z3 form.\ } & 1,08 & 9,3 & - & - \\
\text{Size of z3 form.\ } & 1195684 & 8632341 & - & - \\
\text{Time to solve z3 form.\ }&  0,43 & 8,96 & - & -  \\
\hline
\hline
\multicolumn{5}{|c|}{\text{Reduction} \: \MetFP} \\ 
\hline
\text{Time to build z3 form.\ } & 0.02 & 0,04     & 1,94      & 190,94 \\
\text{Size of z3 form.\ }          & 13268 & 18337 & 1698025 & 139992889 \\
\text{Time to solve z3 form.\ }& 0,1     & 0,23    & 45,51         & - \\
\hline
\hline
\multicolumn{5}{|c|}{\text{Reduction}\:  \MetFPF} \\ 
\hline
\text{Time to build z3 form.\ } & 0,02& 0.04 & 2.07      & 192 \\
\text{Size of z3 form.\ }         &  5180 & 6877 & 551447 & 44643959 \\
\text{Time to solve z3 form.\ }& 0,10 & 35,38 & 180      & - \\
\hline
\hline
\multicolumn{5}{|c|}{\text{Reduction}\: \MetX} \\ 
\hline
\text{Time to build z3 form.\ } & 0.1 & 0.05 & 0.95 & 195 \\
\text{Size of z3 form.\ } & 6285 & 8331 & 726681 & 59 447 253 \\
\text{Time to solve z3 form.\ }& 0,35 & - & - & - \\
\hline
\hline
\end{array}
\]
\end{small}

\subsection{Nim game}

\begin{small}
\[
\begin{array}{|l|c|c|c|c|c|c|c|}
\hline
 & [2,2] &  [3,2] &  [4,5,2]& [3,4,5] &  [2,3,4,4] & [5,4,3,6] &  [2,4,8,14] \\
 \hline 
 \text{\# states} & 16 & 31 & 280 & 328 &  398 & 1595 & 13556 \\
 \text{res} & \text{unsat} & \text{sat}  & \text{sat} & \text{sat} & \text{sat} & \text{sat} &\text{unsat} \\ 
 \hline
\hline
\multicolumn{8}{|c|}{\text{Reduction}\: \MetUU} \\ 
\hline
\text{Time to build z3 form.\ } & 0,00 & 0,00 & 0,09 & 0,31 &  8,99 &  - & - \\
\text{Size of z3 form.\ } & 32 &  105 & 90095 &  285505 & 1011324&   - & -  \\
\text{Time to solve z3 form.\ }& 0,02 &  0,03 & 0,09 & 0,10 &   0,39 &  - &- \\
\hline
\hline
\multicolumn{8}{|c|}{\text{Reduction}\:  \MetFP} \\ 
\hline
\text{Time to build z3 form.\ } & 0,00 &  0,00 & 0,00 & 0,00  & 0,00&  0,19 & 0,21  \\
\text{Size of z3 form.\ } & 104 & 216 & 2283 & 2698 &   3225 & 13865 & 125486 \\
\text{Time to solve z3 form.\ }&  0,01 & 0,02 & 0,04  & 0,4 &  0,06 &  0,15 &  62,2 \\
\hline
\hline
\multicolumn{8}{|c|}{\text{Reduction}\:  \MetFPF} \\ 
\hline
\text{Time to build z3 form.\ } & 0,00 & 0,00 & 0,03 & 0,05 &  0,07 & 0,23 & 6,95 \\
\text{Size of z3 form.\ } &  161 & 326 & 3323 & 3926 & 4703 & 19928 &178158  \\
\text{Time to solve z3 form.\ }& 0,02 & 0,03 & 0,12 & 0,10 & 0,18 &  0,18 &26,66 \\
\hline
\hline\multicolumn{8}{|c|}{\text{Reduction} \: \MetX} \\ 
\hline
\text{Time to build z3 form.\ } & 0,00 & 0,02 & 0,79  & 1,05 & 1,88 & 24,18 &  - \\
\text{Size of z3 form.\ } &  2468 & 9371 & 844412 & 1171832 & 1696982 & 29013387 &  - \\
\text{Time to solve z3 form.\ }& 0,05 & 0,12 & 207,6 &  341,6 & 879,8 &  - &  \\
\hline
\multicolumn{8}{|c|}{\text{Reduction}\: \MetX  \:\text{with bounded Until}} \\ 
\hline
\text{Bound for Until:}            & 10     & 20    & 20        & 20       & 21 & 28 & 43 \\
\text{Time to build z3 form.\ } & 0      & 0,01  & 0,08    & 0,1      & 0,1 &  0,48 & 6,43 \\
\text{Size of z3 form.\ }          & 1682 & 6291 & 65712 & 77816 & 97371 & 540997 & 7311052 \\
\text{Time to solve z3 form.\ }& 0,02  & 0,1    & 2,37    &  3,17   & 4,20  &  184,94 &  - \\
\hline
\hline
\end{array}
\]
\end{small}

\subsection{Resources distribution}

\begin{small}
\[
\begin{array}{|l|c|c|c|c|c|c|}
\hline
n \times m & 10\times5 &  10\times5 &  10\times5 & 10\times7 &  10\times10 & 10\times10 \\
k,d & 2,8 &  4,4  &  4,6 & 4,5 &  2,3 & 4,7 \\
 \hline 
 \text{\# states} & 50 & 50 & 50 & 70 &  100 & 100 \\
 \text{res} & \text{sat} & \text{unsat}  & \text{sat} & \text{unsat} & \text{unsat} & \text{?}  \\ 
 \hline
\hline
\multicolumn{7}{|c|}{\text{Reduction}\: \MetUU - \MetFP} \\ 
\hline
\text{Time to build z3 form.\ } & 0,03 &  0,08 & 0,07 & 0,10  & 0,13 &  0,22   \\
\text{Size of z3 form.\ } & 112155 & 32357 & 70757 & 74417 &  45905  & 283907  \\
\text{Time to solve z3 form.\ }& 0,08  & 274,00 & 0,06  & - &  0,09 &  -  \\
\hline
\hline
\multicolumn{7}{|c|}{\text{Reduction}\:  \MetFPF} \\ 
\hline
\text{Time to build z3 form.\ } & 0,05 &  0,08 & 0,10 & 0,10  & 0,17 &  0,19  \\
\text{Size of z3 form.\ } &  14515 & 23213 & 24715 & 44744 & 43510   & 90916  \\
\text{Time to solve z3 form.\ } & 0,06  & 308,48 & 0,09  & - &  0,10 &  -  \\
\hline
\hline\multicolumn{7}{|c|}{\text{Reduction} \: \MetX} \\ 
\hline    
\text{Time to build z3 form.\ } & 0,08 &  0,02 & 0,03     & 0,03  & 0,02 &  0,09   \\
\text{Size of z3 form.\ }          & 9521 & 7633 & 11133   & 13123 & 7521 & 25733\\
\text{Time to solve z3 form.\ }&  0,05 & 27,16 &  0,02       &  155 &   0,09  &  -  \\
\hline
\hline
\end{array}
\]
\end{small}

\end{document}